\documentclass[a4paper]{article}

\usepackage{amsmath,amssymb,amsthm,mathtools}
\usepackage{algorithm}
\usepackage{algpseudocode}
\usepackage{fullpage}
\usepackage{array}
\usepackage{color}
\usepackage[utf8]{inputenc}
\usepackage[T1]{fontenc}
\usepackage{authblk}              % Standard for author/affiliation formatting
\usepackage[colorlinks=true, linkcolor=blue, urlcolor=blue, citecolor=blue]{hyperref}

\newtheorem{theorem}{Theorem}[section]
\newtheorem{lemma}[theorem]{Lemma}
\newtheorem{proposition}[theorem]{Proposition}
\newtheorem{corollary}[theorem]{Corollary}
\newtheorem{definition}[theorem]{Definition}
\newtheorem{assumption}[theorem]{Assumption}
\newtheorem{remark}[theorem]{Remark}

\newcommand{\cube}{\{-1,1\}^n}
\newcommand{\E}{\mathbb{E}}
\newcommand{\R}{\mathbb{R}}
\newcommand{\TV}{\operatorname{TV}}

\newcommand{\email}[1]{\texttt{\href{mailto:#1}{#1}}}
\newcommand{\authorus}[4]{
    \begin{tabular}{ccc}
        {#1}\\ {\small {#2}}\\ {\footnotesize \email{#3}}
    \end{tabular}
}

\begin{document}

\title{Hierarchical Fourier Approximation for Variational Quantum Distribution Learning}
\author{}
\date{\vspace{-60pt}}
\maketitle
\begin{center}
\authorus{Taha Hoseinpour Asli}{Sharif University of Technology}{tahahoseinpourasli@gmail.com}{1}
\authorus{Sajjad Hashemian}{University of Tehran}{sajjad.hashemian@ut.ac.ir}{2}
\authorus{Ebrahim Ardeshir-Larijani}{Iran University of Science and Technology}{larijani@iust.ac.ir}{3}
\end{center}

\begin{abstract}
We study variational quantum distribution learning through a hierarchy of Walsh--Fourier approximations on the Boolean cube. At each level, a selected set of target Fourier coefficients defines a spectral truncation, which is projected onto the probability simplex and used as the target of a quantum circuit Born machine. Parameters learned at one level initialize the next through a warm-start map. We prove an end-to-end expected learning guarantee where the approximation term is determined by the omitted Fourier mass, while a normalized unbiased estimator yields an explicit statistical bound for empirical truncations. We then instantiate the abstract discrepancy conditions for total variation distance and relate the resulting distributional error to quantum-state fidelity. The total-variation specialization incurs the explicit factor $2^{n-1}$ under our normalized $\ell_2$ convention and is therefore informative only for sufficiently concentrated Fourier tails. The framework does not establish global trainability or eliminate barren plateaus; rather, it identifies the conditions under which low-to-high spectral training admits a approximation--estimation--optimization analysis.
\end{abstract}

\section{Introduction}

In quantum distribution learning, a parameterized quantum circuit is trained to reproduce a target distribution over bit strings. A quantum circuit Born machine (QCBM) supplies the model distribution, while measurements provide estimates of the training objective. A direct fit with a highly expressive circuit may nevertheless face a poorly conditioned landscape, statistical noise, and sensitivity to initialization.

We organize this problem by the spectral resolution of the target distribution. The target is decomposed in the Walsh--Fourier basis and replaced by a nested sequence of increasingly expressive truncations. Training proceeds along this sequence rather than fitting the full target at once. The circuit ansatz remains part of each stage, but the hierarchy itself is determined by which Fourier components of the target are retained. This separation lets us track approximation, optimization, and statistical error as distinct quantities.

The hierarchy should be viewed as an analytical and algorithmic organization of the learning problem. It does not by itself guarantee that a nonconvex variational objective is easy to optimize, nor does it remove measurement noise. Instead, it makes the sources of error explicit. Omitting Fourier coefficients creates approximation error before any circuit is trained, terminating a stage away from its optimum creates optimization error; and replacing population Fourier coefficients by sample estimates creates statistical error. The purpose of the staged construction is to keep these effects separate long enough to analyze each one under its own assumptions.

The closest algorithmic antecedent is the hierarchical QCBM training procedure of Gharibyan, Su, and Tepanyan \cite{gharibyan2023hierarchical}. Their hierarchy begins with the most significant qubits of a bit-string encoding, trains a coarse-grained distribution, and enlarges the active circuit by adding less significant qubits and new parameters. The hierarchy we studied here is different in its organizing variable since we dont need to change the number of qubits, and level $k$ is defined by a nested Fourier support set $\mathcal S_k$. This makes the population approximation error an explicit Fourier-tail quantity. Our analysis should therefore be read as a target-spectrum formulation of staged variational learning, rather than as a claim to have introduced hierarchical circuit training itself.

The contribution lies in how these ingredients are assembled around the target-distribution hierarchy. In particular, we formulates a Fourier-resolution hierarchy for variational distribution learning, in which the population truncation, its simplex projection, the empirical truncation, and the stage QCBM are kept as distinct objects. We then provide and combine Fourier approximation, finite-sample coefficient estimation, finite optimization, and ansatz expressivity in an end-to-end expected learning bound. The expressivity term is retained explicitly and vanishes only under the realizability condition stated in Theorem~\ref{thm:overall}. Lastly, we instantiate the abstract discrepancy conditions for total variation and connects the stage-level distributional error to state fidelity. The factor $2^{n-1}$ is explicit, and the paper records the regime in which this specialization becomes vacuous.
With the ansatz-approximation contribution included inside $E_{\mathrm{optimization}}$ as defined in Theorem~\ref{thm:overall}, the resulting accounting bound has the form
\[
\boxed{
E_{\mathrm{total}}
\le
E_{\mathrm{approx}}
+
E_{\mathrm{optimization}}
+
E_{\mathrm{statistical}}.
}
\]

All optimization conclusions are conditional and local. In particular, the analysis assumes that a stage solution lies in a smooth, strongly convex neighborhood and that the warm-start map places the next stage inside the corresponding neighborhood. These assumptions are not asserted to hold for arbitrary expressive QCBMs. The final theorem should therefore be read as an end-to-end accounting result under stated approximation, optimization, statistical, and expressivity conditions, rather than as a universal trainability theorem.

\subsection{Related Work}
Our approximation analysis uses the standard Walsh--Fourier representation of functions on the Boolean cube \cite{odonnell2014analysis}. The ordering from low to high Fourier degree is also similar to spectral bias in classical neural networks, where lower-frequency components are often learned earlier than higher-frequency ones \cite{rahaman2019spectral}. We use this observation only as motivation. The sets $\mathcal S_k$ are prescribed explicitly, and none of our optimization arguments assumes that a QCBM exhibits spectral bias.

Fourier methods have a different role in studies of data-encoding quantum circuits. In that setting, the circuit output is analyzed as a function of a classical input, and the encoding gates determine which frequencies the model can represent \cite{schuld2021effect,perezsalinas2020data}. Our object of study is instead the Walsh spectrum of the target probability mass function $p:\cube\to\R$. Fourier degree in the target distribution should therefore not be identified with frequencies generated by data encoding or by the circuit parameters. Establishing a useful correspondence between these structures remains an open expressivity question.

Quantum circuit Born machines have been proposed as generative models for classical data and trained with objectives including maximum mean discrepancy \cite{liu2018differentiable,benedetti2019generative}. Subsequent work has examined their implementation and the optimization barriers that arise in quantum generative modeling \cite{rudolph2023trainability}. 

The warm-start component is related to layerwise training and structured initialization of variational circuits \cite{skolik2021layerwise,grant2019initialization}. Gharibyan et al. use such a construction for distribution loading and train the most significant qubits first, append less significant qubits in the state $|+\rangle$, and reuse the learned parameters in the enlarged circuit \cite{gharibyan2023hierarchical}. Their study is primarily algorithmic and numerical. Our hierarchy is organized by retained Walsh coefficients rather than by the number or significance of active qubits. The map $W_k$ transfers parameters between these spectral stages. The gradient-descent analysis itself is standard; what is specific here is its conditional use within the Fourier approximation and statistical-error decomposition \cite{puig2024warm}.

\subsection{Mathematical Preliminaries}

Bit-string measurement outcomes are naturally identified with points of the Boolean cube. The $\{-1,1\}$ convention is particularly convenient because parity functions become the Walsh characters, giving an orthogonal coordinate system for describing correlations of different orders in a distribution.

We work on the Boolean cube $\cube=\{-1,1\}^n$ and denote its probability simplex by
\[
\Delta_{2^n}
=
\left\{
p:\cube\to[0,1]\;:\;\sum_{x\in\cube}p(x)=1
\right\}.
\]
For $S\subseteq[n]$, define the Walsh character
\[
\chi_S(x)=\prod_{i\in S}x_i.
\]
The normalized inner product is
\[
\langle f,g\rangle
=
\E_{x\sim \mathrm{Unif}(\cube)}[f(x)g(x)]
=
2^{-n}\sum_{x\in\cube}f(x)g(x).
\]
The normalization by $2^{-n}$ means that orthogonality is taken with respect to the uniform measure on the cube. We have to stress that the convention is important later and Fourier coefficients are defined using the uniform inner product even when samples are drawn from a nonuniform target distribution $p$. The empirical estimator in Section~\ref{sec:theoretical-analysis} therefore requires an explicit factor of $2^{-n}$.

The characters $\{\chi_S\}_{S\subseteq[n]}$ form an orthonormal basis under this inner product. For $f:\cube\to\R$, we write
\[
\widehat f(S)=\langle f,\chi_S\rangle,
\qquad
f(x)=\sum_{S\subseteq[n]}\widehat f(S)\chi_S(x).
\]
For a probability mass function, $\widehat p(S)$ records its correlation with the parity feature $\chi_S$ under the uniform normalization. Low-cardinality sets $S$ correspond to lower-order interactions, while larger sets encode higher-order parity structure. The hierarchy introduced in the next section selects which of these spectral features are retained at each stage.

Parseval's identity for the Walsh basis is \cite{odonnell2014analysis}
\begin{equation}\label{eq:parseval}
\|f\|_2^2
=
\E_x[f(x)^2]
=
\sum_{S\subseteq[n]}\widehat f(S)^2.
\end{equation}
Parseval's identity converts spectral energy into squared normalized $\ell_2$ norm. It is the basic tool behind the approximation analysis. Once a set of Fourier coefficients is omitted, the sum of their squares is exactly the squared $\ell_2$ norm of the resulting residual.

For $p,q\in\Delta_{2^n}$, we use
\[
\TV(p,q)=\frac12\sum_x |p(x)-q(x)|,
\qquad
\|p-q\|_2=\left(2^{-n}\sum_x(p(x)-q(x))^2\right)^{1/2}.
\]
The normalized $\ell_2$ distance is adapted to Parseval's identity, whereas total variation has the direct interpretation of a discrepancy between classical measurement distributions. Relating these two quantities later introduces a factor that depends exponentially on $n$; this dependence is retained explicitly in all total variation guarantees.

A parameterized $n$-qubit unitary $U(\theta)$, with $\theta\in\Theta\subseteq\R^d$, defines the QCBM distribution
\[
q_\theta(x)
=
|\langle x|U(\theta)|0^n\rangle|^2.
\]
A QCBM is therefore a family of classical distributions obtained by measuring parameterized quantum states in the computational basis, where the circuit parameters are chosen so that these Born probabilities approximate a prescribed target distribution.
A variational distribution-learning objective has the form
\[
L(\theta;p)=D(p,q_\theta),
\]
where $D$ is a discrepancy such as $\TV$, squared $\ell_2$, maximum mean discrepancy, or a regularized KL divergence.

\section{Hierarchical Fourier Approximation}
\label{sec:hfa}

A Fourier hierarchy selects nested sets of coefficients from the expansion above. Each level specifies which spectral features are visible to the learner at that stage, while coefficients outside the active set are deliberately omitted. Moving to a higher level retains a larger part of the target spectrum, so the hierarchy provides a controlled refinement of the target representation rather than a sequence of unrelated learning tasks.

\begin{definition}
We define a Fourier hierarchy as a nested sequence
\(
\mathcal{S}_0\subseteq\mathcal{S}_1\subseteq\cdots\subseteq\mathcal{S}_K\subseteq 2^{[n]}.
\)
The associated truncation operator $T_k$ is
\(
T_k f
=
\sum_{S\in\mathcal{S}_k}\widehat f(S)\chi_S.
\)
The canonical degree hierarchy is $\mathcal{S}_k=\{S\subseteq[n]:|S|\le k\}$.
\end{definition}

The definition allows arbitrary nested spectral sets. The degree hierarchy is the canonical choice because it introduces parity interactions in order of their cardinality, but the analysis uses only nestedness. Thus $T_k$ keeps exactly the coefficients in $\mathcal S_k$, and the choice of hierarchy determines which correlations are represented at level $k$.

We measure the residual in the normalized $\ell_2$ norm. This choice aligns the approximation error with Parseval's identity and makes the spectral contribution of every omitted coefficient explicit.

\begin{definition}
For $p\in\Delta_{2^n}$, define the approximation error
\(
E_{\mathrm{approx}}(k;p)
=
\|p-T_kp\|_2.
\)
\end{definition}

The next theorem is the basic reason for using a Fourier hierarchy. It identifies the approximation error with the Fourier mass omitted from $\mathcal S_k$. Consequently, the error introduced by the hierarchy is not just a modeling term, as $\mathcal S_k$ is chosen, it is determined exactly by the target coefficients outside that set.

\begin{theorem}
\label{thm:parsapprox}
For every $p:\cube\to\R$ and every hierarchy level $k$,
\[
\|p-T_kp\|_2^2
=
\sum_{S\notin\mathcal{S}_k}\widehat p(S)^2.
\]
\end{theorem}
\begin{proof}
The residual has Fourier expansion
\[
p-T_kp=\sum_{S\subseteq[n]}\widehat p(S)\chi_S-\sum_{S\in\mathcal S_k}\widehat p(S)\chi_S=\sum_{S\notin\mathcal S_k}\widehat p(S)\chi_S.
\]
The claim follows immediately from Parseval's identity~\eqref{eq:parseval}; see also \cite{odonnell2014analysis}.
\end{proof}

The proof uses only orthogonality of the Walsh basis. No property of a circuit ansatz or optimization method enters at this point. This separation is useful since the hierarchy determines the best spectral approximation available at level $k$ before the variational model is trained.

Since the hierarchy is nested, $E_{\mathrm{approx}}(k;p)$ is nonincreasing in $k$, and it vanishes when $\mathcal S_K=2^{[n]}$.

The preceding result is an exact $\ell_2$ identity. As there is no need for the truncation $T_kp$ to be a probability distribution, total variation is not applied to it directly. We first record the underlying finite-dimensional norm comparison for arbitrary real-valued functions and then apply it, after simplex projection, to the genuine distributions $p$ and $p_k$.

\begin{figure}[t]
\centering
\includegraphics[width=0.7\linewidth]{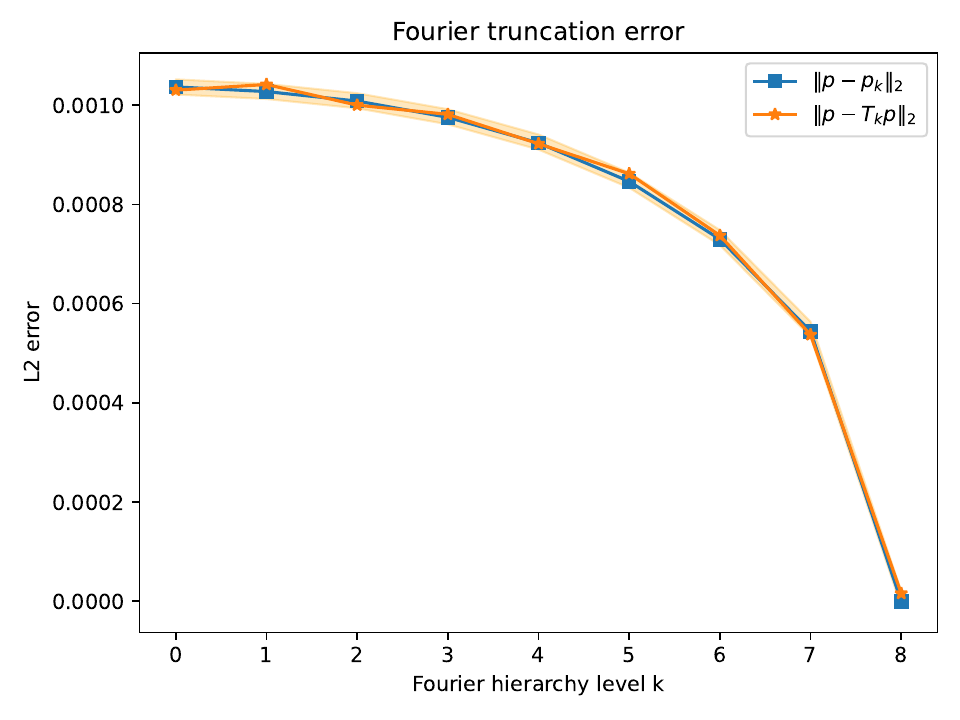}
\caption{The figure tests the implementation of the target hierarchy. The hierarchy level $k$ on the horizontal axis and $\|p_{\delta,K}-T_kp_{\delta,K}\|_2$ on the vertical axis. Show the numerical evaluation and the exact curve for several values of $\delta$, with $n$, $K$, and $\alpha$ fixed. 
}
\label{fig:experimental-truncation}
\end{figure}

\begin{proposition}
\label{prop:tvstab}
For all $f,g:\cube\to\R$,
\[
\frac12\|f-g\|_1
:=
\frac12\sum_{x\in\cube}|f(x)-g(x)|
\le
2^{n-1}\|f-g\|_2.
\]
\end{proposition}
\begin{proof}
This is the standard $\ell_1$--$\ell_2$ comparison on a space of size $2^n$. Cauchy--Schwarz gives $\sum_x|f(x)-g(x)|\le2^n\|f-g\|_2$, and multiplication by $1/2$ yields the claim.
\end{proof}

The dimension-dependent factor is important and in high dimensions, a small normalized $\ell_2$ error does not necessarily yield a useful total variation bound. Section~\ref{sec:qdl} and the Discussion describe the regime in which this conversion is nonvacuous.

Fourier truncation preserves the retained coefficients but need not preserve nonnegativity. Thus $T_kp$ is mathematically convenient as an orthogonal projection, yet it may not itself be a valid distribution for variational training. We therefore project it onto $\Delta_{2^n}$. The following lemma shows that enforcing validity in this way cannot increase the $\ell_2$ distance to the true distribution $p$.

\begin{lemma}[Stability under simplex projection]\label{lem:projstab}
Let $\Pi_\Delta$ be Euclidean projection onto $\Delta_{2^n}$. Then for all $p\in\Delta_{2^n}$ and all $f\in\R^{\cube}$,
\[
\|p-\Pi_\Delta(f)\|_2\le\|p-f\|_2.
\]
In particular, $\|p-\Pi_\Delta(T_kp)\|_2\le\|p-T_kp\|_2$.
\end{lemma}
\begin{proof}
The simplex $\Delta_{2^n}$ is nonempty, closed, and convex. Euclidean projection onto such a set is nonexpansive \cite{bauschke2017convex}. Since $p\in\Delta_{2^n}$, one has $\Pi_\Delta(p)=p$, and therefore
\[
\|p-\Pi_\Delta(f)\|_2
=
\|\Pi_\Delta(p)-\Pi_\Delta(f)\|_2
\le
\|p-f\|_2.
\]
The same inequality holds for the normalized $\ell_2$ norm used here because it is the Euclidean norm multiplied by the constant $2^{-n/2}$.
\end{proof}

This result justifies using $p_k:=\Pi_\Delta(T_kp)$ as the stage target. Combining the proposition and lemma gives
\[
\TV(p,p_k)
\le
2^{n-1}\|p-p_k\|_2
\le
2^{n-1}\|p-T_kp\|_2.
\]
Thus the spectral truncation supplies the approximation structure, while simplex projection restores the probability constraints without worsening the normalized $\ell_2$ guarantee.

\section{Hierarchical Variational Learning}

At level $k$, the projected truncation is the target distribution and the circuit family supplies the model class. The stage therefore asks the QCBM to fit the spectral resolution currently selected by the hierarchy, rather than the full target distribution. Let the stage ansatz be
\[
\Theta_k\subseteq\R^{d_k},
\qquad
q_{k,\theta}(x)=|\langle x|U_k(\theta)|0^n\rangle|^2.
\]
The corresponding target is $p_k=\Pi_\Delta(T_kp)$, the projection of the $k$th Fourier truncation onto the probability simplex. The objects $p$, $T_kp$, and $p_k$ have different roles. $p$ is the true population distribution, $T_kp$ is its possibly nonphysical spectral truncation, and $p_k$ is the valid distribution actually used in the stage objective.

\begin{definition}
For a discrepancy $D$, the $k$th variational learning problem is
\[
\min_{\theta\in\Theta_k} L_k(\theta)
=
D(p_k,q_{k,\theta}).
\]
\end{definition}

The stage objective measures how well the level-$k$ circuit family represents $p_k$. Even if the spectral approximation error is small, a restricted ansatz may fail to represent this target exactly, and a finite optimizer may not reach the best parameter in $\Theta_k$. These effects are accounted for later through the optimization and expressivity terms.

The parameter obtained at level $k$ is transferred to level $k+1$ by a warm-start operator. This map links the stage objectives, the parameters learned at the coarser spectral level provide the initialization for the finer level.

\begin{definition}
A warm-start operator is a map
\[
W_k:\Theta_k\to\Theta_{k+1}
\]
used to initialize level $k+1$ from a solution at level $k$.
\end{definition}

No universal form of $W_k$ is imposed. Its analytical role is captured by the warm-start compatibility assumption in \ref{sec:theoretical-analysis}, which requires the transferred parameter to lie in a neighborhood where the next objective has the stated local geometry.

\begin{remark}[Exact versus empirical hierarchy]\label{rem:exactempirical}
Step 1 of Algorithm~\ref{alg:alg} admits two regimes. In the \emph{oracle} regime, $\widehat p(S)$ is computed exactly, for example from a closed-form description of $p$, and the level-$k$ target is therefore $p_k=\Pi_\Delta(T_kp)$, as defined in the previous section. In the \emph{sample} regime, $\widehat p(S)$ is replaced by the empirical coefficient introduced in Theorem~\ref{thm:stat}. The distribution used for training at level $k$ is then the projected empirical truncation $\widehat p_k:=\Pi_\Delta(\widetilde p_k)$. The overall learning guarantee in Theorem~\ref{thm:overall} is stated for the sample regime, which is relevant to physical experiments with finite measurement budgets. The oracle regime is recovered by setting $E_{\mathrm{statistical}}=0$.
\end{remark}

This distinction prevents population and empirical targets from being conflated. Approximation error is defined using the population truncation, whereas statistical error measures the change caused by replacing its coefficients with sample estimates. The circuit at the final stage is trained against the empirical target in the finite-sample analysis.

Algorithm~\ref{alg:alg} constructs the truncated target at each level, initializes the stage directly or from the preceding solution, and optimizes to the prescribed tolerance. It should be read as a conceptual pipeline. The algorithm specifies which target is formed and how parameters are transferred, but it does not itself guarantee that a practical nonconvex optimizer reaches the requested tolerance; that guarantee is conditional on the local assumptions introduced next.

\begin{figure}
    \centering
    \includegraphics[width=0.9\linewidth]{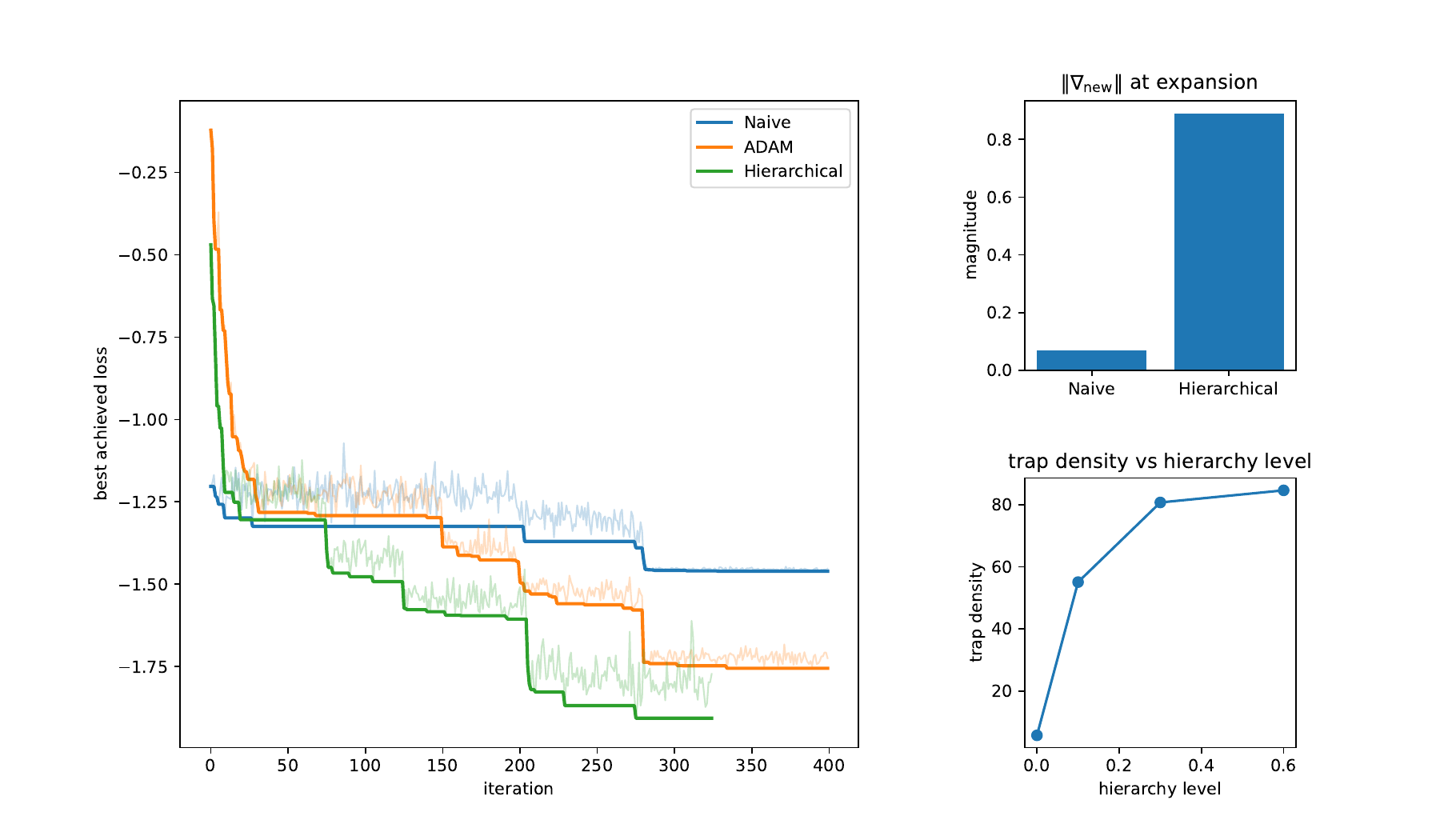}
    \caption{\textbf{(Left)} Convergence of the best achieved loss versus iteration for the Naive, ADAM, and Hierarchical methods. The faint lines represent the raw loss, while the solid lines indicate the cumulative minimum loss. \textbf{(Top Right)} Magnitude of the normalized gradient $\|\nabla_{\text{new}}\|$ at the expansion point, comparing the Naive and Hierarchical approaches. \textbf{(Bottom Right)} The relationship between the hierarchy level and the trap density, showing how the spectral truncation affects the density of the approximation space via the hierarchy level.}
    \label{fig:all}
\end{figure}

\begin{algorithm}[t]
\caption{Hierarchical Fourier Variational Learning}
\label{alg:alg}
\begin{algorithmic}[1]
\Require Samples or oracle access to target distribution $p$; hierarchy $\{\mathcal{S}_k\}_{k=0}^K$; ansatz family $\{U_k(\theta)\}_{k=0}^K$; discrepancy $D$; tolerances $\{\varepsilon_k\}$.
\Ensure Parameter $\theta_K$ defining the final QCBM distribution $q_{K,\theta_K}$.
\For{$k=0,\ldots,K$}
    \State Estimate or compute Fourier coefficients $\widehat p(S)$ for $S\in\mathcal{S}_k$.
    \State Form $T_kp=\sum_{S\in\mathcal{S}_k}\widehat p(S)\chi_S$ and $p_k=\Pi_\Delta(T_kp)$.
    \If{$k=0$}
        \State Choose initialization $\theta_k^{(0)}\in\Theta_k$.
    \Else
        \State Set $\theta_k^{(0)}=W_{k-1}(\theta_{k-1})$.
    \EndIf
    \State Run an optimizer on $L_k(\theta)=D(p_k,q_{k,\theta})$ until $L_k(\theta_k)\le \inf_{\theta\in\Theta_k}L_k(\theta)+\varepsilon_k$.
\EndFor
\State \Return $\theta_K$.
\end{algorithmic}
\end{algorithm}

\section{Theoretical Analysis}
\label{sec:theoretical-analysis}
The analysis keeps the three error sources separate. Optimization is controlled locally at each hierarchy level; sampling affects the empirical Fourier truncation; the remaining term is the population truncation error from Section~\ref{sec:hfa}. This separation matters because the three terms arise from different mechanisms and require different hypotheses. The approximation term is fixed by the target spectrum and the chosen hierarchy, the optimization term depends on the circuit family and training procedure, and the statistical term depends on the finite sample used to construct the empirical target. We combine them only after deriving the optimization and statistical bounds.

The optimization argument applies in a neighborhood of a stage minimizer and does not assert global geometry for the variational objective. This is a deliberate limitation where expressive variational quantum objectives are generally nonconvex, and we can't claim that the hierarchy changes their global landscape. The assumption below identifies only the region in which the stage-wise convergence analysis is valid.

\begin{assumption}[Local optimization geometry]\label{ass:geom}
For each $k$, the objective $L_k$ is $\mu_k$-strongly convex and $\beta_k$-smooth on a convex neighborhood $\mathcal{N}_k\subseteq\Theta_k$ containing a minimizer $\theta_k^\star$.
\end{assumption}

Strong convexity supplies a local error-growth condition around $\theta_k^\star$, while smoothness controls the behavior of a gradient step. Together they yield the standard linear contraction used later, but only as long as the iterates remain in $\mathcal N_k$.

A sufficiently accurate solution must also transfer into the neighborhood used at the next level. This is the specific role of the hierarchy in the optimization analysis where the previous stage is useful only if its output is a valid initialization for the local model at the finer stage.

\begin{assumption}[Warm-start compatibility]\label{ass:warm}
For each $k<K$,
\[
W_k(\theta_k)\in\mathcal{N}_{k+1}
\]
whenever
\[
L_k(\theta_k)-L_k(\theta_k^\star)\le \varepsilon_k.
\]
\end{assumption}

Warm-start compatibility is therefore not a statement that consecutive objectives are globally similar. It asserts only that once stage $k$ is solved to tolerance $\varepsilon_k$, the lifted parameter lies inside $\mathcal N_{k+1}$. The additional trajectory condition in the lemma ensures that subsequent gradient steps do not leave the region where the local inequalities apply.

These assumptions place the warm-started trajectory in the region where smoothness and strong convexity give a linear contraction. The next lemma records that contraction explicitly.

\begin{lemma}[Warm-start optimization]\label{lem:warmstart}
Assume local optimization geometry and warm-start compatibility, and assume that the gradient-descent trajectory remains in $\mathcal N_{k+1}$. This is ensured, for example, when $\mathcal N_{k+1}$ contains the sublevel set determined by the warm-started objective value. Gradient descent with step size $1/\beta_{k+1}$ initialized at $W_k(\theta_k)$ satisfies
\[
L_{k+1}(\theta^{(t)})-L_{k+1}(\theta_{k+1}^\star)
\le
\left(1-\frac{\mu_{k+1}}{\beta_{k+1}}\right)^t
\left(
L_{k+1}(W_k(\theta_k))-L_{k+1}(\theta_{k+1}^\star)
\right).
\]
\end{lemma}
\begin{proof}
Write $L=L_{k+1}$, $\beta=\beta_{k+1}$, $\mu=\mu_{k+1}$, and $\theta^\star=\theta_{k+1}^\star$. For a $\beta$-smooth, $\mu$-strongly convex function, the standard descent and Polyak--{\L}ojasiewicz inequalities give \cite{nesterov2018lectures}
\[
L\left(\theta-\frac1\beta\nabla L(\theta)\right)
\le L(\theta)-\frac1{2\beta}\|\nabla L(\theta)\|_2^2,
\qquad
\|\nabla L(\theta)\|_2^2\ge2\mu\bigl(L(\theta)-L(\theta^\star)\bigr).
\]
Combining them yields the one-step contraction
\[
L(\theta^+)-L(\theta^\star)\le\Bigl(1-\frac\mu\beta\Bigr)\bigl(L(\theta)-L(\theta^\star)\bigr).
\]
Assumption~\ref{ass:warm} places $\theta^{(0)}=W_k(\theta_k)$ in $\mathcal N_{k+1}$ once stage $k$ terminates. The trajectory assumption allows the contraction to be iterated $t$ times, giving
\[
L(\theta^{(t)})-L(\theta^\star)\le\Bigl(1-\frac\mu\beta\Bigr)^t\bigl(L(\theta^{(0)})-L(\theta^\star)\bigr)=\Bigl(1-\frac{\mu_{k+1}}{\beta_{k+1}}\Bigr)^t\bigl(L_{k+1}(W_k(\theta_k))-L_{k+1}(\theta_{k+1}^\star)\bigr).
\]
\end{proof}

The proof follows the standard local gradient-descent argument. Smoothness bounds the objective decrease produced by a step of size $1/\beta_{k+1}$, and strong convexity implies the Polyak--{\L}ojasiewicz inequality relating gradient norm to suboptimality. Combining the two gives a one-step contraction, while warm-start compatibility and trajectory containment justify applying the same contraction repeatedly at the new hierarchy level.

Solving the contraction bound for the iteration count gives the stage-wise optimization complexity. The rate depends on the local condition number $\beta_k/\mu_k$ and only logarithmically on the ratio between the initial suboptimality and the target tolerance.

\begin{theorem}[Optimization complexity]
Under the preceding assumptions, level $k$ reaches optimization error at most $\varepsilon_k$ after
\[
t_k
\ge
\frac{\beta_k}{\mu_k}
\log
\frac{L_k(\theta_k^{(0)})-L_k(\theta_k^\star)}{\varepsilon_k}
\]
gradient steps.
\end{theorem}
\begin{proof}
Let $\Delta_k=L_k(\theta_k^{(0)})-L_k(\theta_k^\star)$. Lemma~\ref{lem:warmstart}, or the same contraction applied directly when $k=0$, gives
\[
L_k(\theta^{(t)})-L_k(\theta_k^\star)\le
\Bigl(1-\frac{\mu_k}{\beta_k}\Bigr)^t\Delta_k
\le e^{-t\mu_k/\beta_k}\Delta_k,
\]
where the last step uses $1-a\le e^{-a}$. Solving $e^{-t\mu_k/\beta_k}\Delta_k\le\varepsilon_k$ for $t$ gives
\[
t\ge\frac{\beta_k}{\mu_k}\log\frac{\Delta_k}{\varepsilon_k},
\]
with the bound already satisfied at $t=0$ when $\Delta_k\le\varepsilon_k$.
\end{proof}

This theorem does not establish that an arbitrary optimizer enters the local neighborhood or that the local constants are favorable. It quantifies the number of gradient steps required once the assumptions hold and the stated initialization is available.

In the sample regime, the truncated target is built from Fourier coefficients estimated using independent samples from $p$. A normalization issue is essential here. The coefficient $\widehat p(S)$ is defined using the uniform-measure inner product, but the observations $X_i$ are distributed according to $p$. Consequently, the raw empirical average of $\chi_S(X_i)$ estimates $2^n\widehat p(S)$ rather than $\widehat p(S)$, and the estimator must include the factor $2^{-n}$.

The theorem below quantifies the resulting statistical error. Orthogonality ensures that the squared $\ell_2$ error of the empirical truncation is the sum of the squared coefficient errors, so its expectation is the sum of the corresponding variances.

\begin{theorem}[Statistical estimation of Fourier coefficients]\label{thm:stat}
Let $X_1,\ldots,X_m$ be independent samples from $p\in\Delta_{2^n}$. Define the empirical Fourier coefficient
\[
\widehat{\widetilde p}(S):=\frac{2^{-n}}{m}\sum_{i=1}^m\chi_S(X_i),
\qquad
\widetilde p_k(x)
:=
\sum_{S\in\mathcal{S}_k}
\widehat{\widetilde p}(S)\,\chi_S(x).
\]
Then $\widehat{\widetilde p}(S)$ is an unbiased estimator of $\widehat p(S)$ for every $S\in\mathcal S_k$, and
\[
\E\|T_kp-\widetilde p_k\|_2^2
=
\frac{2^{-2n}}{m}\sum_{S\in\mathcal{S}_k}\operatorname{Var}_{X\sim p}[\chi_S(X)]
\le
\frac{|\mathcal{S}_k|}{2^{2n}m}.
\]
\end{theorem}
\begin{proof}
Since $X_i\sim p$ and $\widehat p(S)=2^{-n}\sum_xp(x)\chi_S(x)$,
\[
\E[\chi_S(X_i)]=\sum_xp(x)\chi_S(x)=2^n\widehat p(S),
\qquad
\E[\widehat{\widetilde p}(S)]=\widehat p(S).
\]
Independence gives
\[
\operatorname{Var}(\widehat{\widetilde p}(S))=\frac{2^{-2n}}{m^2}\sum_{i=1}^m\operatorname{Var}[\chi_S(X_i)]=\frac{2^{-2n}}{m}\operatorname{Var}_{X\sim p}[\chi_S(X)].
\]
Parseval's identity~\eqref{eq:parseval} and unbiasedness then yield
\[
\E\|T_kp-\widetilde p_k\|_2^2=\sum_{S\in\mathcal S_k}\operatorname{Var}(\widehat{\widetilde p}(S))=\frac{2^{-2n}}{m}\sum_{S\in\mathcal S_k}\operatorname{Var}_{X\sim p}[\chi_S(X)]\le\frac{2^{-2n}|\mathcal S_k|}{m}.
\]
The final inequality uses $\chi_S(X)\in\{-1,1\}$ and hence $\operatorname{Var}[\chi_S(X)]\le1$.
\end{proof}

The proof separates into the normalization, variance, and Parseval steps. The $2^{-n}$ factor makes each retained coefficient unbiased; independence gives the usual $1/m$ variance reduction; and Parseval converts the coefficientwise variances into expected squared function error. Under the present normalization, the upper bound scales linearly with the number $|\mathcal S_k|$ of retained coefficients and inversely with the sample size $m$.

The same estimate controls the projected empirical target used by the algorithm. Indeed, nonexpansiveness of simplex projection gives
\[
\|p_k-\widehat p_k\|_2
=
\|\Pi_\Delta(T_kp)-\Pi_\Delta(\widetilde p_k)\|_2
\le
\|T_kp-\widetilde p_k\|_2.
\]
Consequently,
\[
\E\|p_k-\widehat p_k\|_2^2
\le
\frac{|\mathcal S_k|}{2^{2n}m},
\qquad
\E\|p_k-\widehat p_k\|_2
\le
\frac{\sqrt{|\mathcal S_k|}}{2^n\sqrt m}.
\]
These are the projected statistical quantities used in Theorem~\ref{thm:overall}.

For a target statistical accuracy $\varepsilon_k^{\mathrm{stat}}>0$, the sufficient sample size
\[
m\ge\frac{|\mathcal S_k|}{2^{2n}\,(\varepsilon_k^{\mathrm{stat}})^2},
\]
implies $\E\|T_kp-\widetilde p_k\|_2^2\le(\varepsilon_k^{\mathrm{stat}})^2$ and, by Jensen's inequality,
\[
\E\|T_kp-\widetilde p_k\|_2\le\bigl(\E\|T_kp-\widetilde p_k\|_2^2\bigr)^{1/2}\le\varepsilon_k^{\mathrm{stat}}.
\]
This is a sufficient expectation bound, not a high-probability statement. It describes how many samples are enough to control the retained empirical Fourier expansion in mean normalized $\ell_2$ error.

\begin{figure}
    \centering
    \includegraphics[width=0.75\linewidth]{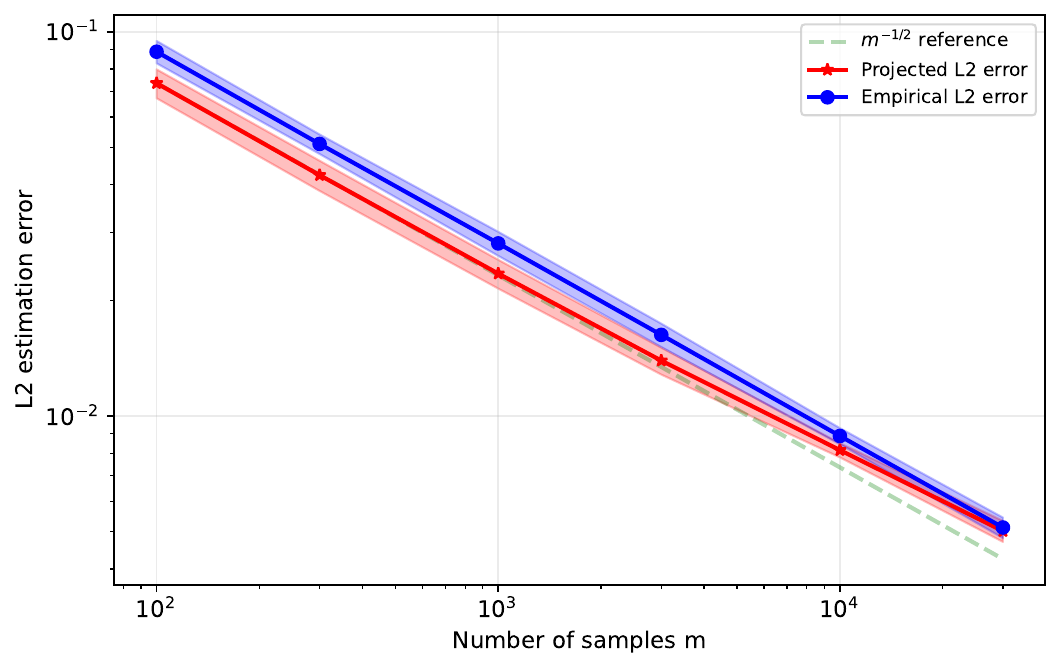}
    \caption{
Mean normalized $\ell_2$ estimation error is shown as a function of the sample size $m$ on logarithmic axes, with $n$, $K$, $k$, $\alpha$, and $\delta$ held fixed. We compare the raw truncation error, $\|T_kp-\widetilde p_k\|_2$, with the error after projection onto the probability simplex, $\|T_kp-\Pi_\Delta(\widetilde p_k)\|_2$. shaded intervals denote confidence intervals over independent data samples. The $m^{-1/2}$ reference line highlights the expected Monte Carlo scaling.
    }
    \label{fig:experimental-estimation}
\end{figure}

The final decomposition requires the stage discrepancy to compare approximation, sampling, and training errors through a common quantity. The $\ell_2$-domination property converts the spectral and empirical truncation errors into $D$, while the triangle inequality allows intermediate distributions to be inserted between the true target and the circuit output.

\begin{assumption}[Discrepancy regularity]\label{ass:discrepancy}
The discrepancy $D:\Delta_{2^n}\times\Delta_{2^n}\to\R_{\ge0}$ used in the stage objectives satisfies, for some constant $C_D<\infty$ and all $a,b,c\in\Delta_{2^n}$:
\begin{enumerate}
\item[(i)] ($\ell_2$-domination) $D(a,b)\le C_D\|a-b\|_2$;
\item[(ii)] (triangle inequality) $D(a,c)\le D(a,b)+D(b,c)$.
\end{enumerate}
\end{assumption}

This assumption is not intended to cover every discrepancy used in quantum machine learning. In particular, Section~\ref{sec:qdl} verifies it for total variation distance with an explicit constant. Other objectives, such as KL divergence, may require a different formulation because the same triangle-inequality argument is unavailable.

The population truncation $p_K$, the empirical truncation $\widehat p_K$, and the learned circuit distribution $q_{K,\widehat\theta}$ are distinct objects. The population target $p_K$ records only spectral approximation; $\widehat p_K$ adds the error from estimating retained coefficients; and $q_{K,\widehat\theta}$ adds the effects of finite optimization and limited ansatz expressivity. Comparing them in that order exposes the three contributions.

The next theorem is the central synthesis of the paper. It inserts $p_K$ and $\widehat p_K$ as intermediate distributions between $p$ and the learned model, then assigns each comparison to one error term.

\begin{theorem}[Overall learning guarantee]\label{thm:overall}
Let $p\in\Delta_{2^n}$. Let $p_K=\Pi_\Delta(T_Kp)$ denote the population-level Fourier truncation at the top hierarchy level, and let $\widehat p_K=\Pi_\Delta(\widetilde p_K)$ denote its empirical counterpart built from $m$ i.i.d.\ samples from $p$ as in Theorem~\ref{thm:stat}. Let $\widehat\theta\in\Theta_K$ satisfy the near-optimality guarantee
\[
D(\widehat p_K,q_{K,\widehat\theta})\le\inf_{\theta\in\Theta_K}D(\widehat p_K,q_{K,\theta})+\varepsilon_{\mathrm{opt}}.
\]
If $D$ satisfies the discrepancy regularity assumption (Assumption~\ref{ass:discrepancy}) with constant $C_D$, then
\[
\E\bigl[D(p,q_{K,\widehat\theta})\bigr]
\le
E_{\mathrm{approx}}
+
E_{\mathrm{optimization}}
+
E_{\mathrm{statistical}},
\]
where
\[
E_{\mathrm{approx}}=C_D\|p-p_K\|_2,
\qquad
E_{\mathrm{statistical}}=C_D\,\E\|p_K-\widehat p_K\|_2,
\qquad
E_{\mathrm{optimization}}=\varepsilon_{\mathrm{opt}}+\E\Big[\inf_{\theta\in\Theta_K}D(\widehat p_K,q_{K,\theta})\Big].
\]
In particular, if the ansatz at level $K$ is expressive enough to exactly represent $\widehat p_K$ for some $\theta\in\Theta_K$ (e.g.\ a universal QCBM family), then $\inf_\theta D(\widehat p_K,q_{K,\theta})=0$ for every realization of the sample, hence $\E[\inf_\theta D(\widehat p_K,q_{K,\theta})]=0$ and $E_{\mathrm{optimization}}=\varepsilon_{\mathrm{opt}}$, recovering exactly the boxed decomposition of the Introduction.
\end{theorem}
\begin{proof}
Apply the triangle inequality, part (ii) of Assumption~\ref{ass:discrepancy}, twice using the intermediate distributions $p_K$ and $\widehat p_K$:
\[
D(p,q_{K,\widehat\theta})\le D(p,p_K)+D(p_K,\widehat p_K)+D(\widehat p_K,q_{K,\widehat\theta}).
\]
By the $\ell_2$-domination property in part (i), $D(p,p_K)\le C_D\|p-p_K\|_2=E_{\mathrm{approx}}$. This term is deterministic and depends only on the population distribution $p$. Similarly, $D(p_K,\widehat p_K)\le C_D\|p_K-\widehat p_K\|_2$. Taking expectations over the sampling randomness defining $\widehat p_K$ gives $\E[D(p_K,\widehat p_K)]\le C_D\,\E\|p_K-\widehat p_K\|_2=E_{\mathrm{statistical}}$. The near-optimality hypothesis holds pathwise for every sample realization:
\[
D(\widehat p_K,q_{K,\widehat\theta})\le\inf_{\theta\in\Theta_K} D(\widehat p_K,q_{K,\theta})+\varepsilon_{\mathrm{opt}}.
\]
Taking expectations of this pointwise inequality,
\[
\E\big[D(\widehat p_K,q_{K,\widehat\theta})\big]\le\E\Big[\inf_{\theta\in\Theta_K} D(\widehat p_K,q_{K,\theta})\Big]+\varepsilon_{\mathrm{opt}}=E_{\mathrm{optimization}}.
\]
Combining the optimization bound with the deterministic approximation bound $D(p,p_K)\le E_{\mathrm{approx}}$ and the statistical bound $\E[D(p_K,\widehat p_K)]\le E_{\mathrm{statistical}}$ gives
\[
\E\bigl[D(p,q_{K,\widehat\theta})\bigr]\le E_{\mathrm{approx}}+E_{\mathrm{statistical}}+E_{\mathrm{optimization}},
\]
as claimed.
\end{proof}

The proof is a two-step triangle-inequality argument. The first comparison, from $p$ to $p_K$, is deterministic and measures the population approximation error. The second, from $p_K$ to $\widehat p_K$, is random and measures the sampling error. The final comparison is controlled by the near-optimality condition and retains the infimum over the ansatz family. This expressivity term is essential and disappears only under the additional realizability condition stated in the theorem.

\paragraph{Computational cost.}
The hierarchy also determines the number of Fourier coefficients that must be estimated and stored. For the degree hierarchy, this number is the count of subsets of size at most $k$.

For the degree hierarchy $\mathcal{S}_k=\{S:|S|\le k\}$, counting subsets by cardinality gives
\(
|\mathcal{S}_k|=\sum_{j=0}^k \binom{n}{j}.
\)
If evaluating $\chi_S(X)$ is charged according to the number $|S|$ of factors in the parity, coefficient estimation at level $k$ costs
\(
O\left(m_k\sum_{S\in\mathcal S_k}|S|\right).
\)
For the degree hierarchy this is at most $O(km_k|\mathcal S_k|)$. Thus, if optimization uses $t_k$ iterations with per-iteration circuit cost $C_k$, a direct implementation has total training cost
\[
O\left(
\sum_{k=0}^K
\left(
k m_k|\mathcal{S}_k|+t_kC_k
\right)
\right),
\]
and coefficient storage is
\[
O\left(\max_{0\le k\le K}|\mathcal{S}_k|\right).
\]
The first term counts direct parity evaluation without assuming constant-time access to every $\chi_S(X)$, the second counts the variational iterations, and nested coefficient tables can be updated incrementally so that only the largest table must be stored. If parity values are precomputed or maintained incrementally, the coefficient-evaluation term may be reduced to $O(m_k|\mathcal S_k|)$.

This bookkeeping bound does not assert practical scalability. Both $|\mathcal S_k|$ and the circuit cost $C_k$ may grow rapidly, and the optimization iteration count depends on local conditioning. Its purpose is to make the separate estimation and variational contributions explicit.

\section{Quantum Distribution Loading}
\label{sec:qdl}

The projected truncation $p_k$ is a valid classical distribution and therefore defines an amplitude-encoded quantum state. At hierarchy level $k$, this state is the quantum-loading target associated with the currently retained Fourier information:
\[
|\psi_k\rangle
=
\sum_{x\in\cube}\sqrt{p_k(x)}|x\rangle,
\]
provided $p_k$ is a valid distribution. The QCBM prepares a state with computational-basis measurement distribution $q_k$, which is fitted to $p_k$. The original distribution $p$ remains the final target; $p_k$ is the spectral approximation used at this stage.

Fitting $q_k$ to $p_k$ controls only the stage-level loading problem. It does not remove the discrepancy between the spectral target $p_k$ and the original distribution $p$. The triangle inequality for total variation separates these two effects into the Fourier approximation term between $p$ and $p_k$ and the loading term between $p_k$ and $q_k$.

\begin{corollary}
If a circuit prepares $|\phi_k\rangle$ whose measurement distribution is $q_k$ and $\TV(p_k,q_k)\le\varepsilon_k$, then
\[
\TV(p,q_k)
\le
\TV(p,p_k)+\varepsilon_k
\le
2^{n-1}\|p-p_k\|_2+\varepsilon_k.
\]
\end{corollary}
\begin{proof}
The first inequality is the triangle inequality for total variation, and the second is Proposition~\ref{prop:tvstab} applied to $p$ and $p_k$.
\end{proof}

The corollary makes the role of the hierarchy explicit in the loading task. Even perfect preparation of the level-$k$ target would leave the approximation error $\TV(p,p_k)$. Conversely, improving the Fourier level does not control the circuit error $\TV(p_k,q_k)$ unless the variational stage is also trained accurately.

Assumption~\ref{ass:discrepancy}, used in Theorem~\ref{thm:overall}, asks for $\ell_2$-domination and a triangle inequality. Total variation has both properties, so the abstract bound specializes to the loading problem. The specialization is exact at the level of the assumptions, but it introduces the dimension-dependent constant $C_D=2^{n-1}$ inherited from the normalized $\ell_2$ comparison.

The next statement therefore serves two purposes. It verifies that total variation is a legitimate choice for the abstract discrepancy $D$, and it displays the resulting approximation and statistical terms without suppressing their dependence on dimension.

\begin{remark}
\label{thm:realize}
The discrepancy $D=\TV$ satisfies Assumption~\ref{ass:discrepancy} on $\Delta_{2^n}$ with constant $C_D=2^{n-1}$, for all $a,b,c\in\Delta_{2^n}$,
\[
\TV(a,b)\le2^{n-1}\|a-b\|_2
\qquad\text{and}\qquad
\TV(a,c)\le\TV(a,b)+\TV(b,c).
\]
Consequently, taking $D=\TV$ in Theorem~\ref{thm:overall} is legitimate with no further hypotheses on $D$, and the overall learning bound for the quantum-loading pipeline of this section takes the fully explicit form
\[
\E\bigl[\TV(p,q_{K,\widehat\theta})\bigr]
\le
2^{n-1}\|p-p_K\|_2
+
2^{n-1}\,\E\|p_K-\widehat p_K\|_2
+
E_{\mathrm{optimization}}.
\]
The $\ell_2$-domination bound is Proposition~\ref{prop:tvstab}, and total variation satisfies the standard triangle inequality. Substituting $C_D=2^{n-1}$ into Theorem~\ref{thm:overall} gives the displayed bound.
\end{remark}

No additional learning argument is required in this statement and the result is an instantiation of the abstract decomposition. Its importance can interpreted as it connects the error terms derived in normalized $\ell_2$ to a distributional discrepancy relevant for computational-basis measurement outcomes, while also exposing the exponential factor that limits the usefulness of the bound.

The preceding theorem concerns measurement distributions. At the state level, the target and prepared QCBM states can instead be compared by fidelity. Computational-basis measurement is a quantum channel, so data processing transfers a trace-distance bound between states to a total variation bound between their measurement distributions.

\begin{proposition}
\label{prop:fidelity}
Let $|\psi_k\rangle=\sum_x\sqrt{p_k(x)}|x\rangle$ and $|\phi_k\rangle=U_k(\theta)|0^n\rangle$ be the target and prepared amplitude-encoded states, with respective computational-basis measurement distributions $p_k$ and $q_k$. Then
\[
\TV(p_k,q_k)\le\sqrt{1-F(\psi_k,\phi_k)^2},
\qquad
F(\psi_k,\phi_k):=|\langle\psi_k|\phi_k\rangle|.
\]
\end{proposition}
\begin{proof}
Trace distance is contractive under computational-basis measurement, and for pure states it equals $\sqrt{1-F(\psi_k,\phi_k)^2}$ \cite{watrous2018theory}. Applying these standard facts to $|\psi_k\rangle\langle\psi_k|$ and $|\phi_k\rangle\langle\phi_k|$ gives the claim.
\end{proof}

The proof uses only standard quantum-information inequalities. Contractivity under measurement and the pure-state relation between trace distance and fidelity. Thus a small fidelity loss controls $\TV(p_k,q_k)$; the loading corollary then controls $\TV(p,q_k)$, and Theorem~\ref{thm:realize} inserts the same discrepancy into the overall three-term bound. Fidelity control does not remove the Fourier approximation term, but it supplies a state-level route to controlling the variational loading term.

The approximation term $E_{\mathrm{approx}} = C_D \|p - p_K\|_2$ is small only when the omitted Fourier mass is small. Fourier concentration is therefore an assumption on the target family, not a universal property of distributions on the Boolean cube. The definition below isolates the targets for which the chosen hierarchy retains enough spectral mass to make the approximation theorem informative.

\begin{definition}[Fourier-concentrated distributions]
For $\epsilon>0$, let $\mathcal{P}_{K, \epsilon}$ be the family of \emph{$(\epsilon, K)$-concentrated distributions} $p \in \Delta_{2^n}$ satisfying
\[
\sum_{S \notin \mathcal{S}_K} \widehat p(S)^2 \le \left(\frac{\epsilon}{C_D}\right)^2.
\]
\end{definition}
By Theorem~\ref{thm:parsapprox}, every $p \in \mathcal{P}_{K, \epsilon}$ satisfies the hierarchical approximation bound $E_{\mathrm{approx}} \le \epsilon$. In other words, the concentration condition is chosen precisely so that the omitted Fourier energy contributes at most $\epsilon$ after conversion through the discrepancy constant $C_D$.

The explicit example used in the appendix is a product distribution with independent biased bits. For such a target, the coefficient indexed by $S$ is the product of the corresponding single-bit biases, together with the uniform normalization: $\widehat{p}(S) = 2^{-n} \prod_{i \in S} m_i$. When the biases are bounded away from one, multiplying more of them suppresses higher-order coefficients, giving exponential decay with $|S|$. On $\mathcal{P}_{K, \epsilon}$, Theorem~\ref{thm:overall} converts this assumed spectral concentration into the stated end-to-end learning guarantee.

\section{Discussion}

We introduced a hierarchy of variational learning problems obtained by retaining progressively larger sets of Walsh--Fourier coefficients of the target distribution. The main guarantee separates the final error into Fourier approximation, statistical estimation, and optimization terms, with ansatz expressivity included explicitly in the optimization term. The approximation error is determined by the omitted Fourier mass, the statistical error depends on the number of retained coefficients and the sample size, and the optimization bound follows under the stated local geometry and warm-start conditions. For quantum distribution loading, total variation provides a concrete instance of the abstract discrepancy assumption, while Proposition~\ref{prop:fidelity} relates the distributional error to state fidelity.

These results are conditional. The analysis does not establish global convergence, rule out barren plateaus, or guarantee that an arbitrary QCBM enters the required local neighborhood. It also does not assume that the circuit can represent the empirical target exactly; any failure of the ansatz to do so remains in the term $\inf_{\theta\in\Theta_K}D(\widehat p_K,q_{K,\theta})$. The method is therefore most useful for target families whose Fourier mass is concentrated on the retained sets and for circuit families whose local optimization geometry is compatible with the warm-start map.

The total variation specialization has an additional limitation. Under the normalized $\ell_2$ convention used here, the conversion constant is $2^{n-1}$. A small Fourier-tail error in $\ell_2$ may therefore give a weak or vacuous total variation bound when $n$ is large. The Parseval approximation result remains exact, but a useful distributional guarantee requires the approximation and statistical errors to be small enough to offset this factor.

Several questions remain open. One is to identify circuit families whose expressivity can be related directly to the retained Fourier support. Another is to replace the local convexity assumptions by conditions suited to nonconvex variational objectives. Adaptive choices of the hierarchy and extensions to discrepancies such as maximum mean discrepancy or KL divergence also require separate analyses, since these losses need not satisfy the regularity assumptions used in Theorem~\ref{thm:overall}.

\bibliographystyle{alpha}
\bibliography{references}

@book{odonnell2014analysis,
  author    = {Ryan O'Donnell},
  title     = {Analysis of Boolean Functions},
  publisher = {Cambridge University Press},
  year      = {2014}
}

@book{nesterov2018lectures,
  author    = {Yurii Nesterov},
  title     = {Lectures on Convex Optimization},
  publisher = {Springer},
  year      = {2018}
}

@book{bauschke2017convex,
  author    = {Heinz H. Bauschke and Patrick L. Combettes},
  title     = {Convex Analysis and Monotone Operator Theory in Hilbert Spaces},
  edition   = {2},
  publisher = {Springer},
  year      = {2017}
}

@book{watrous2018theory,
  author    = {John Watrous},
  title     = {The Theory of Quantum Information},
  publisher = {Cambridge University Press},
  year      = {2018}
}

@article{gharibyan2023hierarchical,
  author  = {Hrant Gharibyan and Vincent Su and Hayk Tepanyan},
  title   = {Hierarchical Learning for Quantum {ML}: Novel Training Technique for Large-Scale Variational Quantum Circuits},
  journal = {arXiv preprint arXiv:2311.12929},
  year    = {2023}
}

@article{liu2018differentiable,
  author  = {Jin-Guo Liu and Lei Wang},
  title   = {Differentiable Learning of Quantum Circuit Born Machine},
  journal = {Physical Review A},
  volume  = {98},
  number  = {6},
  pages   = {062324},
  year    = {2018}
}

@article{benedetti2019generative,
  author  = {Marcello Benedetti and Delfina Garcia-Pintos and Oscar Perdomo and Vicente Leyton-Ortega and Yunseong Nam and Alejandro Perdomo-Ortiz},
  title   = {A Generative Modeling Approach for Benchmarking and Training Shallow Quantum Circuits},
  journal = {npj Quantum Information},
  volume  = {5},
  pages   = {45},
  year    = {2019}
}

@article{rudolph2023trainability,
  author  = {Manuel S. Rudolph and Sacha Lerch and Supanut Thanasilp and Oriel Kiss and Sofia Vallecorsa and Michele Grossi and Zo{\"e} Holmes},
  title   = {Trainability Barriers and Opportunities in Quantum Generative Modeling},
  journal = {arXiv preprint arXiv:2305.02881},
  year    = {2023}
}

@article{skolik2021layerwise,
  author  = {Andrea Skolik and Jarrod R. McClean and Masoud Mohseni and Patrick van der Smagt and Martin Leib},
  title   = {Layerwise Learning for Quantum Neural Networks},
  journal = {Quantum Machine Intelligence},
  volume  = {3},
  pages   = {5},
  year    = {2021}
}

@article{puig2024warm,
  author  = {Ricard Puig and Marc Drudis and Supanut Thanasilp and Zo{\"e} Holmes},
  title   = {Variational Quantum Simulation: A Case Study for Understanding Warm Starts},
  journal = {arXiv preprint arXiv:2404.10044},
  year    = {2024}
}

@article{grant2019initialization,
  author  = {Edward Grant and Leonard Wossnig and Mateusz Ostaszewski and Marcello Benedetti},
  title   = {An Initialization Strategy for Addressing Barren Plateaus in Parametrized Quantum Circuits},
  journal = {Quantum},
  volume  = {3},
  pages   = {214},
  year    = {2019}
}

@inproceedings{rahaman2019spectral,
  author    = {Nasim Rahaman and Aristide Baratin and Devansh Arpit and Felix Draxler and Min Lin and Fred A. Hamprecht and Yoshua Bengio and Aaron Courville},
  title     = {On the Spectral Bias of Neural Networks},
  booktitle = {Proceedings of the 36th International Conference on Machine Learning},
  year      = {2019}
}

@article{schuld2021effect,
  author  = {Maria Schuld and Ryan Sweke and Johannes Jakob Meyer},
  title   = {Effect of Data Encoding on the Expressive Power of Variational Quantum-Machine-Learning Models},
  journal = {Physical Review A},
  volume  = {103},
  number  = {3},
  pages   = {032430},
  year    = {2021}
}

@article{perezsalinas2020data,
  author  = {Adri{\'a}n P{\'e}rez-Salinas and Alba Cervera-Lierta and Elies Gil-Fuster and Jos{\'e} I. Latorre},
  title   = {Data Re-uploading for a Universal Quantum Classifier},
  journal = {Quantum},
  volume  = {4},
  pages   = {226},
  year    = {2020}
}

\appendix

\section*{Code Availability}
Our implementation of the hierarchical Fourier approximations as a python module for Pennylane and the experimental results provided in this manuscript are open-source and available at the following \href{https://github.com/sajjadhashemian/Fourier-Hierarchical-VQC/}{github} repository: \url{https://github.com/sajjadhashemian/Fourier-Hierarchical-VQC/}.

\section{Realization of Optimization Assumptions for Product Distributions}

We verify the local optimization assumptions for a product target and an unentangled product ansatz. To make the two stages consistent with the full objective below, this appendix uses the two-level hierarchy
\[
\mathcal S_0=\{\varnothing\},
\qquad
\mathcal S_1=2^{[n]}.
\]
Thus $p_0$ is uniform and $p_1=p$. This appendix does not assert that the same constants hold for the canonical degree-one truncation.

\subsection{Target Distribution, Ansatz, and Exact Objective}

\begin{definition}[Axis-aligned product distribution]
Let the target distribution on $\cube=\{-1,1\}^n$ be
\[
p(x)=\prod_{i=1}^n\frac{1+m_ix_i}{2},
\qquad
m_i\in[0,1-\epsilon],
\qquad
\epsilon\in(0,1).
\]
\end{definition}

The unentangled ansatz is
\[
U(\theta)=\bigotimes_{i=1}^nR_y(\theta_i),
\qquad
R_y(\theta_i)=\exp\left(-i\frac{\theta_i}{2}Y_i\right),
\]
and, with $x_i=1$ corresponding to computational-basis outcome $0$ and $x_i=-1$ to outcome $1$, its Born distribution is
\[
q_\theta(x)=\prod_{i=1}^n\frac{1+\cos(\theta_i)x_i}{2}.
\]
Write
\[
m_S:=\prod_{i\in S}m_i,
\qquad
c_i:=\cos(\theta_i),
\qquad
c_S:=\prod_{i\in S}c_i,
\]
with empty products equal to one. Direct expansion in the Walsh basis gives
\[
\widehat p(S)=2^{-n}m_S,
\qquad
\widehat q_\theta(S)=2^{-n}c_S.
\]
Consequently, Parseval's identity gives the exact stage-one objective
\begin{equation}\label{eq:appendix-exact-loss}
L_1(\theta)
:=
\|p-q_\theta\|_2^2
=
2^{-2n}\sum_{S\subseteq[n]}(m_S-c_S)^2.
\end{equation}
Equivalently,
\begin{equation}\label{eq:appendix-product-loss}
L_1(\theta)
=
2^{-2n}\left[
\prod_{i=1}^n(1+m_i^2)
+
\prod_{i=1}^n(1+c_i^2)
-
2\prod_{i=1}^n(1+m_ic_i)
\right].
\end{equation}
In particular, the higher-degree terms are parameter dependent; the full loss is not a sum of independent one-coordinate losses.

For later reference, define
\[
A_i(c):=\prod_{j\ne i}(1+c_j^2),
\qquad
B_i(c):=\prod_{j\ne i}(1+m_jc_j),
\]
and, for $i\ne j$,
\[
A_{ij}(c):=\prod_{\ell\ne i,j}(1+c_\ell^2),
\qquad
B_{ij}(c):=\prod_{\ell\ne i,j}(1+m_\ell c_\ell).
\]
Differentiating~\eqref{eq:appendix-product-loss} yields
\begin{equation}\label{eq:appendix-gradient}
\frac{\partial L_1}{\partial\theta_i}
=
2^{1-2n}\sin(\theta_i)
\left[m_iB_i(c)-c_iA_i(c)\right].
\end{equation}
The exact Hessian entries are
\begin{align}
\frac{\partial^2L_1}{\partial\theta_i^2}
&=
2^{1-2n}
\left[
(1-2c_i^2)A_i(c)+m_ic_iB_i(c)
\right],
\label{eq:appendix-hessian-diagonal}
\\
\frac{\partial^2L_1}{\partial\theta_i\partial\theta_j}
&=
2^{1-2n}\sin(\theta_i)\sin(\theta_j)
\left[
2c_ic_jA_{ij}(c)-m_im_jB_{ij}(c)
\right]
\quad(i\ne j).
\label{eq:appendix-hessian-offdiagonal}
\end{align}

The unique optimum in $[0,\pi/2]^n$ is
\[
\theta^*:=(\arccos m_1,\ldots,\arccos m_n),
\]
for which $q_{\theta^*}=p$ and $L_1(\theta^*)=0$. At this point,
\begin{align}
[\nabla^2L_1(\theta^*)]_{ii}
&=
2^{1-2n}(1-m_i^2)\prod_{\ell\ne i}(1+m_\ell^2),
\label{eq:appendix-opt-hessian-diagonal}
\\
[\nabla^2L_1(\theta^*)]_{ij}
&=
2^{1-2n}
\sqrt{1-m_i^2}\sqrt{1-m_j^2}\,m_im_j
\prod_{\ell\ne i,j}(1+m_\ell^2)
\quad(i\ne j).
\label{eq:appendix-opt-hessian-offdiagonal}
\end{align}

\subsection{Uniform Local Strong Convexity and Smoothness}

Set
\[
\gamma:=2\epsilon-\epsilon^2,
\qquad
M:=10n^{3/2}2^{-n},
\qquad
r:=\frac{\gamma\,2^{-n}}{10n^{3/2}},
\]
and define the convex neighborhood
\[
\mathcal N_1
:=
\{\theta\in\R^n:\|\theta-\theta^*\|_2\le r\}.
\]

\begin{theorem}[Local Strong Convexity and Smoothness]\label{thm:appendix-geometry}
On $\mathcal N_1$, the exact objective~\eqref{eq:appendix-exact-loss} satisfies
\[
\mu_1I
\preceq
\nabla^2L_1(\theta)
\preceq
\beta_1I,
\]
with the valid explicit choices
\[
\mu_1:=2^{-2n}\gamma,
\qquad
\beta_1:=n2^{-n}+2^{-2n}\gamma.
\]
Thus Assumption~\ref{ass:geom} holds on $\mathcal N_1$.
\end{theorem}

\begin{proof}
From~\eqref{eq:appendix-exact-loss}, the Hessian at the zero-residual point $\theta^*$ is a Gram matrix:
\[
\nabla^2L_1(\theta^*)
=
2^{1-2n}
\sum_{S\subseteq[n]}
\nabla c_S(\theta^*)\nabla c_S(\theta^*)^{\mathsf T}.
\]
The singleton terms $S=\{i\}$ imply
\[
\nabla^2L_1(\theta^*)
\succeq
2^{1-2n}\operatorname{diag}(1-m_1^2,\ldots,1-m_n^2)
\succeq
2^{1-2n}\gamma I.
\]
Moreover, $\|\nabla c_S(\theta^*)\|_2^2\le|S|$, so
\[
\lambda_{\max}(\nabla^2L_1(\theta^*))
\le
2^{1-2n}\sum_{S\subseteq[n]}|S|
=
n2^{-n}.
\]

It remains to make these pointwise bounds uniform. Every partial derivative of $c_S$ of order at most three has absolute value at most one. Since $|m_S-c_S|\le2$, every third partial derivative of $(m_S-c_S)^2$ has absolute value at most $10$. Summing the $2^n$ Fourier terms in~\eqref{eq:appendix-exact-loss} therefore gives
\[
\max_{i,j,k}
\left|
\frac{\partial^3L_1}{\partial\theta_i\partial\theta_j\partial\theta_k}
\right|
\le
10\,2^{-n}.
\]
Hence, for all $\theta,\vartheta\in\R^n$,
\[
\|\nabla^2L_1(\theta)-\nabla^2L_1(\vartheta)\|_{\mathrm{op}}
\le
M\|\theta-\vartheta\|_2.
\]
For $\theta\in\mathcal N_1$, the definition of $r$ gives
\[
Mr
=
2^{-2n}\gamma
=
\frac12\,2^{1-2n}\gamma.
\]
Weyl's eigenvalue inequality now yields
\[
\lambda_{\min}(\nabla^2L_1(\theta))
\ge
2^{1-2n}\gamma-Mr
=
2^{-2n}\gamma
=
\mu_1
\]
and
\[
\lambda_{\max}(\nabla^2L_1(\theta))
\le
n2^{-n}+Mr
=
n2^{-n}+2^{-2n}\gamma
=
\beta_1.
\]
These bounds hold uniformly over all of $\mathcal N_1$.
\end{proof}

\subsection{Warm Start and Trajectory Containment}

The stage-zero target is uniform, and its product-circuit optimum is
\[
\theta^{(0)}
=
(\pi/2,\ldots,\pi/2),
\qquad
q_{\theta^{(0)}}(x)=2^{-n}.
\]
The distance from this point to the stage-one optimum is
\[
\|\theta^{(0)}-\theta^*\|_2
=
\left(\sum_{i=1}^n\arcsin(m_i)^2\right)^{1/2}.
\]
The original condition $m_i\in[0,1-\epsilon]$ alone does not force this distance to be smaller than $r$. The following theorem therefore states a sufficient additional condition for this concrete warm start to lie in the verified local neighborhood. The radius is deliberately conservative: it is obtained from a global Hessian-Lipschitz bound so that all constants remain explicit.

\begin{theorem}[Warm-start Compatibility and Containment]\label{thm:appendix-warm}
Assume, in addition to the target definition, that
\begin{equation}\label{eq:appendix-small-bias}
\left(\sum_{i=1}^n\arcsin(m_i)^2\right)^{1/2}
\le
\frac r2.
\end{equation}
Let $W_0$ be the identity map and initialize stage one at $W_0(\theta^{(0)})=\theta^{(0)}$. Gradient descent with step size $1/\beta_1$ then remains in $\mathcal N_1$ and satisfies
\[
\|\theta^{(t)}-\theta^*\|_2
\le
\left(1-\frac{\mu_1}{\beta_1}\right)^t
\|\theta^{(0)}-\theta^*\|_2.
\]
Consequently, Assumption~\ref{ass:warm} and the trajectory-containment condition of Lemma~\ref{lem:warmstart} both hold for this example. The objective error obeys
\[
L_1(\theta^{(t)})-L_1(\theta^*)
\le
\left(1-\frac{\mu_1}{\beta_1}\right)^t
\bigl(L_1(\theta^{(0)})-L_1(\theta^*)\bigr),
\]
and it is at most $\varepsilon_1$ after
\[
t
\ge
\frac{\beta_1}{\mu_1}
\log\frac{L_1(\theta^{(0)})-L_1(\theta^*)}{\varepsilon_1}
\]
iterations.
\end{theorem}

\begin{proof}
Condition~\eqref{eq:appendix-small-bias} gives
\[
\|W_0(\theta^{(0)})-\theta^*\|_2
\le r/2<r,
\]
so the warm start belongs to $\mathcal N_1$. Suppose inductively that $\theta^{(t)}\in\mathcal N_1$. The segment from $\theta^*$ to $\theta^{(t)}$ is contained in the convex set $\mathcal N_1$, and the fundamental theorem of calculus gives
\[
\nabla L_1(\theta^{(t)})
=
H_t(\theta^{(t)}-\theta^*),
\qquad
H_t
:=
\int_0^1
\nabla^2L_1\bigl(\theta^*+s(\theta^{(t)}-\theta^*)\bigr)\,ds.
\]
Theorem~\ref{thm:appendix-geometry} implies $\mu_1I\preceq H_t\preceq\beta_1I$. Therefore
\begin{align*}
\|\theta^{(t+1)}-\theta^*\|_2
&=
\left\|
\left(I-\frac1{\beta_1}H_t\right)
(\theta^{(t)}-\theta^*)
\right\|_2
\\
&\le
\left(1-\frac{\mu_1}{\beta_1}\right)
\|\theta^{(t)}-\theta^*\|_2.
\end{align*}
Thus every iterate remains within distance $r/2$ of $\theta^*$, proving trajectory containment. The objective contraction and iteration bound then follow from Lemma~\ref{lem:warmstart} and the iteration bound following it.
\end{proof}

Without the additional small-bias condition~\eqref{eq:appendix-small-bias}, the local geometry theorem remains valid around $\theta^*$, but the proposed uniform warm start is not guaranteed to enter that neighborhood. No convergence claim from $\theta^{(0)}=(\pi/2,\ldots,\pi/2)$ follows from the local theory in that case.

\end{document}